\documentclass[letterpaper,10pt,conference]{IEEEtran}

\usepackage{amsmath,amssymb, amsthm,cite,graphicx}
\usepackage{mathrsfs}

\long\def\symbolfootnote[#1]#2{\begingroup
\def\thefootnote{\fnsymbol{footnote}}\footnote[#1]{#2}\endgroup}
\newtheorem{theorem}{Theorem}

\newcommand{\In}{\mathsf{In}}
\newcommand{\Mid}{\mathsf{Mid}}
\newcommand{\Out}{\mathsf{Out}}
\newcommand{\DC}{\mathsf{DC}}
\newcommand{\Source}{\mathsf{S}}

\title{Exact Minimum-Repair-Bandwidth Cooperative Regenerating Codes for Distributed Storage Systems}
\author{Kenneth W. Shum and Yuchong Hu \\ Institute of Network Coding, \\ The Chinese University of Hong Kong. \\
Email: \texttt{\{wkshum,ychu\}@inc.cuhk.edu.hk}.
}

\begin{document}

\maketitle

\begin{abstract} In order to provide high data reliability, distributed storage systems disperse data with redundancy to multiple storage nodes. Regenerating codes is a new class of erasure codes to introduce redundancy
for the purpose of  improving the data repair performance in distributed storage. Most of the studies on regenerating codes focus on the single-failure recovery, but it is not uncommon to see two or more node failures at the same time in large storage networks. To exploit the opportunity of repairing multiple failed nodes simultaneously,  a cooperative repair mechanism, in the sense that the nodes to be repaired can exchange data among themselves, is investigated. A lower bound on the repair-bandwidth for cooperative repair is derived and a construction of a family of exact cooperative regenerating codes matching this lower bound is presented.
\symbolfootnote[0]{This work was partially supported by a grant from the University Grants Committee of the Hong Kong Special Administrative Region, China (Project No. AoE/E-02/08).}
\end{abstract}


\section{Introduction}

Distributed storage systems such as Oceanstore~\cite{Oceanstore} and Total Recall~\cite{Totalrecall} provide reliable and scalable solutions to the increasing demand of data storage. They distribute data with redundancy to multiple storage nodes and the data can be retrieved even if some of nodes are not available. When erasure coding is used as a redundancy scheme in distributed storage, the task of repairing a node failure becomes non-trivial. A traditional way to repair a failed node is to download and reconstruct the whole data file first, and then regenerate the lost content (e.g., RAID-5, RAID-6). Since the size of the original data file may be huge, a lot of traffic is consumed for the purpose of repairing just one failed node.

In order to reduce the total traffic required for repairing, called {\em repair-bandwidth}, a new class of erasure codes, called {\em regenerating codes}~\cite{DGWR2010}, is presented and has a significantly lower traffic consumed in regenerating a failed node. The main idea of regenerating codes is to reduce repair-bandwidth from the survival nodes to a new node (called a {\em newcomer}), which regenerates the lost content in the failed node. Some constructions of minimum repair-bandwidth regenerating codes are given in~\cite{RSKR09, SK11}. They are based on {\em exact repair} or called {\em exact} MBR codes, which means the lost content of the failed node are repaired exactly.

Most of the studies on regenerating codes in the literature are for the single-failure recovery or one-by-one repair. When the number of storage nodes becomes large, the multi-failure case is not infrequent, and we need to regenerate several failed nodes at the same time. In addition, in practical systems such as Total Recall, a recovery process is triggered only after the total number of failed nodes has reached a predefined threshold.  These facts motivates the regeneration of multiple failed nodes jointly, instead of repairing in a one-by-one manner. A repair process in which the newcomers may exchange packets among themselves, called a {\em cooperative repair} or {\em cooperative recovery}, is first introduced in~\cite{HXWZL10}. We will call the regenerating codes for multiple failures with cooperative repair {\em cooperative regenerating codes}. In~\cite{WXHO10}, a special class of cooperative regenerating codes is proposed, in which the newcomers can select survival nodes for repairing flexibly. In~\cite{Shum11}, an explicit construction of cooperative regenerating code minimizing the storage in each node is given.

The tradeoff spectrum between repair-bandwidth and storage for cooperative regenerating codes is given in~\cite{Shum11, KSS}.
Regenerating codes which attain one end of this spectrum, corresponding to the minimum storage,  are considered in~\cite{HXWZL10, WXHO10}.
In this paper, we focus on the other end of this spectrum. Codes which minimizes repair-bandwidth is called {\em Minimal Repair-Bandwidth Cooperative Regenerating} (MBCR) codes.

{\em Main Results:} After presenting a simple example and demonstrating the basic ideas in Section~\ref{sec:example}, we derive in Section~\ref{sec:LB} a lower bound on the repair-bandwidth in cooperative recovery.  An explicit construction of a family of exact MBCR codes matching this lower bound is given in Section~\ref{sec:explicit}.


\section{An Illustrative Example} \label{sec:example}
In this section, we  introduce  some notations and illustrate the basic idea of cooperative repair.

Based on the system model introduced in~\cite{DGWR2010} and~\cite{HXWZL10}, a file consisting of $B$ packets is encoded and distributed to $n$ nodes and a data collector can retrieve the file by downloading data from any $k$ of $n$ nodes. When $r$ nodes fails, $r$ newcomers are selected to repair the failed nodes. The repair process is divided into two phases. In the first phase, each of the $r$ newcomers connects to $d$ surviving nodes and downloads some packets. In the second phase, the newcomers exchange some
packets among themselves. The objective is to minimize the total number of the packets transmitted (i.e., repair-bandwidth) in the two phases. Next we give an illustration of cooperative repair with parameters $n=4$ and $d=k=r=2$.

We initialize the distributed storage system by dividing a data file into eight data packets $A$, $B,\ldots, H$, and distribute them to four storage nodes. Each node stores five packets: four systematic and one parity-check (Fig.~\ref{fig:CR}). The addition ``+'' is bit-wise exclusive-OR (XOR). The first node stores the first four packets $A$, $B$, $C$, $D$, skips the packet $E$, and stores the sum of the next two packets, $F+G$. The content of nodes 2,  3 and 4 can be obtained likewise by shifting the encoding pattern to the right respectively by 2, 4 and 6 packets. It is easy to verify that a data collector can rebuild the file from any two of four nodes in the illustrated code. For example, the data collector which connects to nodes 1 and 2 can reconstruct the eight data packets by downloading $A$, $B$, $C$ and $F+G$ from node 1, and $D$, $E$, $F$, and $H+A$ from node 2. Then it can solve for $G$ by subtracting $F$ from $F+G$, and $H$ by subtracting $A$ from $H+A$.

\begin{figure}
\centering
\includegraphics[width=2.9in]{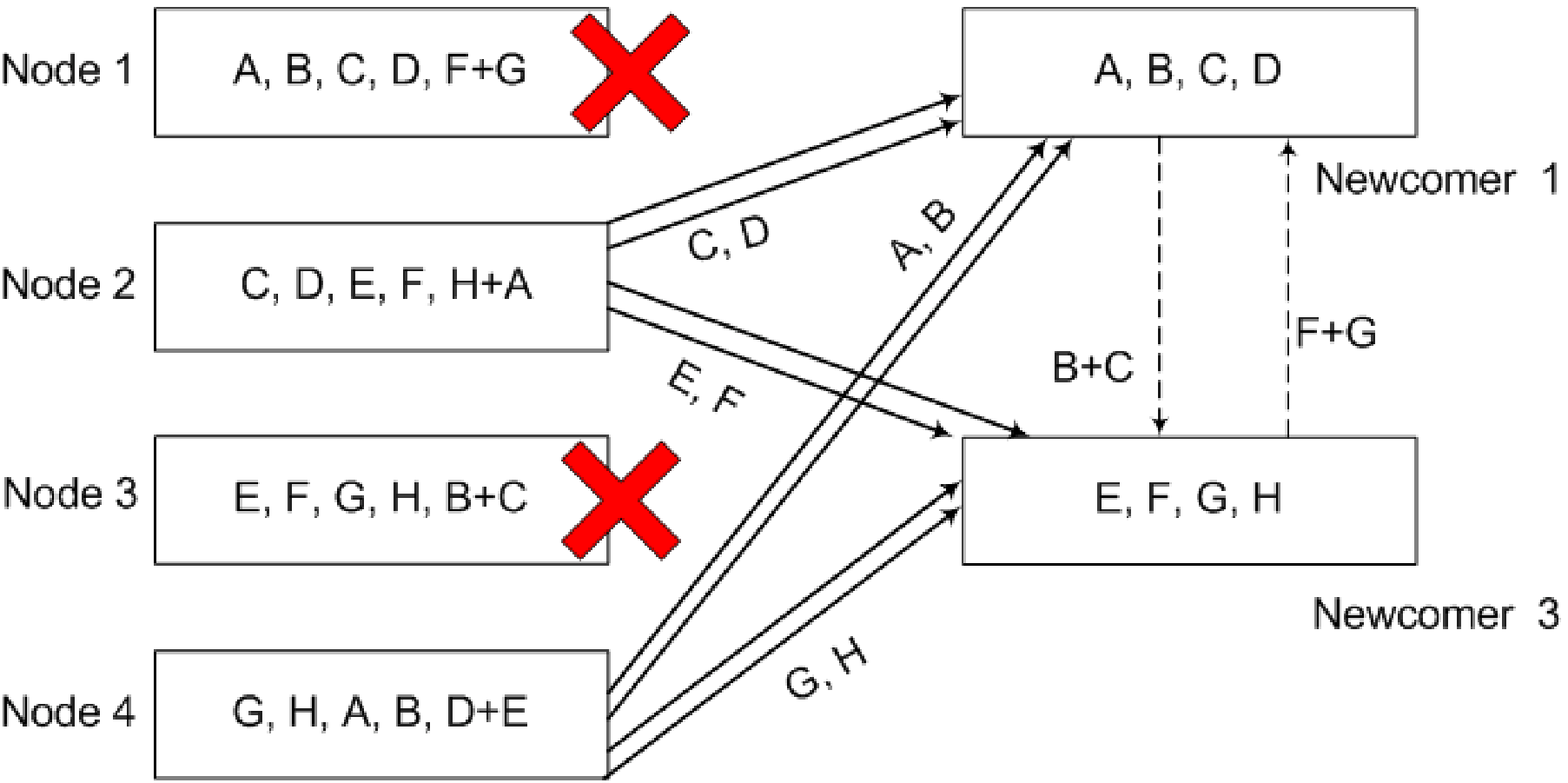}

\vspace{0.8cm}

\includegraphics[width=2.9in]{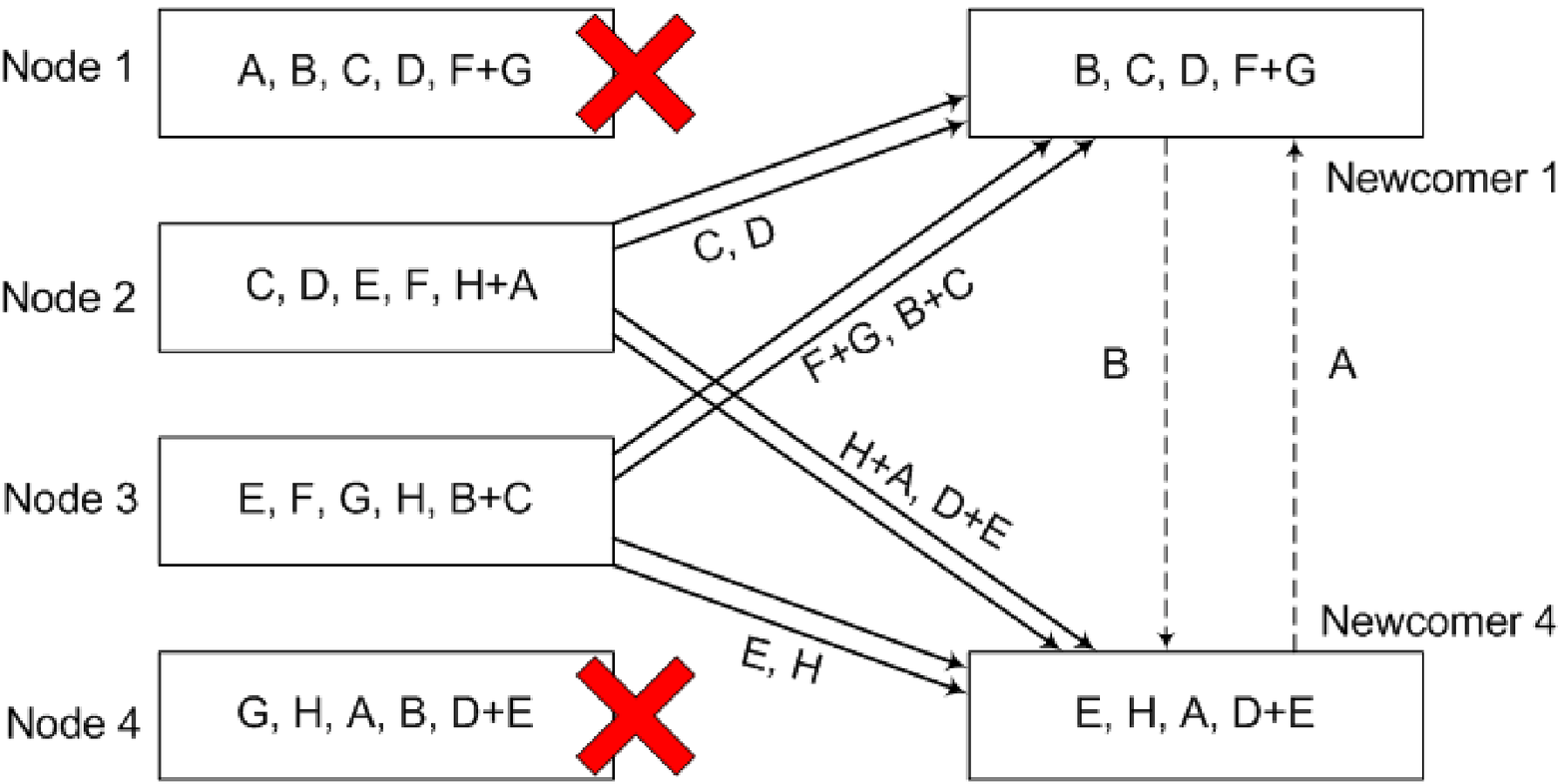}
\caption{An example of cooperative repair. The labels of the solid (resp. dashed) arrows indicate the packets transmitted during the first (resp. second) phase of the repair process. The content of the newcomers after the first phase of the repair process is shown.} \label{fig:CR}
\end{figure}

As for the repair process, the illustrated code costs ten packets per any two-failure recovery. Suppose that nodes 1 and 3 fail (see the first diagram in Fig.~\ref{fig:CR}). Both newcomers 1 and 3 first download four packets from the survival nodes 2 and~4. Then newcomer 1 (resp. newcomer 3) computes the sum $B+C$ (resp. $F+G$) and sends it to newcomer~3 (resp. newcomer~1). Obviously, a total of ten packets, which are equal to the number of lines (including solid and dashed lines), are transmitted in the network. Similarly, should nodes 2 and 4 fail, the same repair-bandwidth is consumed for regeneration. Suppose that node 1 and 4 fail (see the second diagram in Fig.~\ref{fig:CR}). Both newcomers 1 and 4 first download four packets from the survival nodes 2 and 3. Note that among the four downloaded packets, newcomer 1 (resp. newcomer 4) receives one encode packet $F+G$ (resp. $D+E$) from node 3 (resp. node 2). Then newcomer 1 solves for packet $B$ by subtracting $C$ from $B+C$, and transmits packet $B$ to newcomer~4. Also, newcomer 4 solves for packet $A$ and sends it to newcomer~1. Clearly, a total of ten packet transmissions are sufficient. Similarly, if any pair of two adjacent storage nodes fail, we can also repair them with ten packet transmissions, using the symmetry in the encoding for data distribution.

\section{Lower Bound on Repair-Bandwidth for Multi-Loss Cooperative Repair} \label{sec:LB}
In this paper, we assume that the storage nodes are symmetrical; for the storage cost,  each node stores $\alpha$ packets, and for the repair-bandwidth, each newcomer connects to $d$  existing nodes and downloads $\beta_1$ packets from each of them, and then sends $\beta_2$ packets to each of the $r-1$ other newcomers. In this paper, we only consider the case that $d\geq k$. The repair-bandwidth per newcomer, denoted by $\gamma$, is defined as the total number of the packets each newcomer receives, and thus is equal to
$$ \gamma = d\beta_1 + (r-1)\beta_2.$$
The aim of this section is to derive a lower bound on $\gamma$.

To formulate the problem, we draw an information flow graph as in~\cite{HXWZL10}.
Given parameters $n$, $k$, $d$, $r$, $\alpha$, $\beta_1$ and $\beta_2$, we construct an information flow graph $G=(\mathcal{V},\mathcal{E})$ as follows. The vertices are grouped into stages.

\begin{itemize}
\item In stage $-1$, there is only one vertex $\Source$, representing the source node which has the original file.

\item In stage 0, there are $n$ vertices $\Out_1$, $\Out_2,\ldots, \Out_n$, each of them corresponds to an initial storage node. There is a directed edge with capacity $\alpha$ from $\Source$ to each $\Out_i$.

\item For $t=1,2,3,\ldots$, suppose $r$ nodes fail in stage~$t$. Let the indices of these $r$ storage nodes be $\mathcal{S}_t = \{j_1$, $j_2, \ldots, j_r\}$. For each $i \in \mathcal{S}_t$, we put three vertices $\In_i$, $\Mid_i$ and $\Out_i$ in stage $t$. There are $d$ directed edges, with capacity $\beta_1$ from $d$ ``out'' vertices in previous stages to each $\In_i$. For each $i\in \mathcal{S}_t$, we put a directed edge from $\In_i$ to $\Mid_i$  with infinite capacity, and a directed edge from $\Mid_i$ to $\Out_i$ with capacity $\alpha$. For each pair of distinct indices $i,j \in\mathcal{S}_t$, we draw a directed edge from $\In_i$ to $\Mid_j$ with capacity~$\beta_2$. The edges with capacity $\beta_1$ represent the data transferred from existing storage nodes to newcomers, and the edges with capacity $\beta_2$ represent the data exchanged among the newcomers.

\item To a data collector, who shows up after $s$ repair processes have taken place, we put a vertex $\DC$ in stage $s$ and connect it with $k$ ``out'' vertices with distinct indices in stage $s$ or earlier. The capacities of these $k$ edges are set to infinity.
\end{itemize}
An example of information flow graph for $n=4$ storage nodes and $d=r=k=2$ is shown in Fig.~\ref{fig:flowgraph1}.

\begin{figure}
\centering
\includegraphics[width=2.4in]{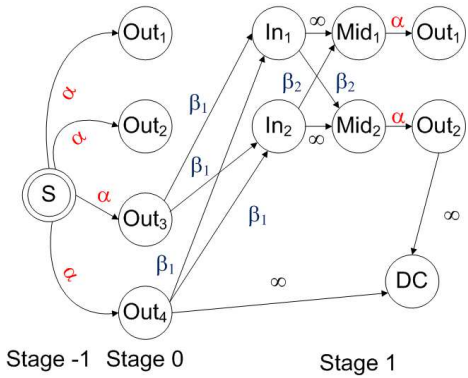}
\caption{An example of information flow graph.} \label{fig:flowgraph1}
\end{figure}

To derive a lower bound on $d\beta_1 + (r-1)\beta_2$ for a file of fixed size $B$ is equivalent to derive a upper bound of $B$ for given capacities $\beta_1$ and $\beta_2$ in the information flow graph.
So we can apply a celebrated {\em max-flow} theorem in~\cite{ACLY00}, which says the size of the data file $B$ cannot be larger than the max-flow from $\mathsf{S}$ to any data collector ($\DC$). The max-flow is the maximal value of all feasible flows from $\mathsf{S}$ to $\DC$. Here, a {\em flow} from $\mathsf{S}$ to $\DC$, called an $(\mathsf{S},\DC)$-{\em flow}, is a mapping $F$ from the set of edges to the set of non-negative real numbers, satisfying (i) for every edge $e$,  $F(e)$ does not exceed the capacity of $e$, (ii) for any vertex $v$ except the source vertex $\mathsf{S}$ and the terminal vertex $\DC$, the sum of $F(e)$ over edges $e$ terminating at $v$ is equal to the sum of $F(e)$ over edges $e$ going out from $v$,
\[
 \sum_{e=(u,v)\in\mathcal{E}} F(e) =
 \sum_{e=(v,u)\in\mathcal{E}} F(e).
\]
The {\em value} of an $(\mathsf{S},\DC)$-flow $F$ is defined as
\[
  \sum_{e=(u,\DC)\in\mathcal{E}} F(e).
\]

From the max-flow-min-cut theorem, we can upper bound the value of a flow by the capacity of a {\em cut}. For a given data collector $\DC$, an $(\mathsf{S},\DC)$-cut is a partition $(\mathcal{U}, \bar{\mathcal{U}})$ of the vertices in the information flow graph, such that $\Source \in \mathcal{U}$ and $\DC \in \bar{\mathcal{U}}$, where $\bar{\mathcal{U}}$ stands for the complement of $\mathcal{U}$ in the vertex set $\mathcal{V}$. The {\em capacity} of an $(\mathsf{S},\DC)$-cut is defined as the sum of capacities of the edges from vertices in $\mathcal{U}$ to vertices in $\bar{\mathcal{U}}$.
Next, we will use the fact that the value of any $(\mathsf{S},\DC)$-{flow} is less than or equal to the capacity of any $(\mathsf{S},\DC)$-cut,  together with the max-flow theorem in~\cite{ACLY00}, to prove the following theorem.

\begin{theorem}
If $d \geq k $, the repair-bandwidth $ d \beta_1 + (r-1)\beta_2$ is lower bounded by
\begin{equation}
\frac{B(2d+r-1)}{k(2d+r-k)} \label{eq:maxflow}
\end{equation}
and this lower bound can be met only when
\begin{equation}
 (\beta_1,\beta_2) = \Big(\frac{2B}{k(2d+r-k)}, \frac{B}{k(2d+r-k)}\Big).
 \label{eq:vertex}
\end{equation}
 \label{thm:maxflow}
\end{theorem}


\begin{proof}
Consider a data collector which downloads data from $k$ out of $n$ storage nodes.
By re-labeling the storage nodes, we can assume without loss of generality that the corresponding $k$ ``out'' vertices be $\Out_1$, $\Out_2,\ldots, \Out_{k}$.
Suppose that these $k$ ``out'' vertices belong to stages 1 to $s$, and for $\nu=1,2,\ldots, s$, there are $\ell_\nu$ ``out'' nodes in $\bar{\mathcal{U}}$ in stage~$\nu$. By vertex re-labeling, we can assume without loss of generality that $\Out_1$, $\Out_2,\ldots, \Out_{\ell_1}$ belong to stage $t_1$, $\Out_{\ell_1+1}$, $\Out_{\ell_1+2},\ldots, \Out_{\ell_1+\ell_2}$ belong to stage $t_2$, and so on. For notational convenience, we let $\ell_0=0$.
Consider the $(\mathsf{S},\DC)$-cut with $\bar{\mathcal{U}}$ consisting of the vertices
$$\bigcup_{j=\ell_0+\ell_1+\ldots+\ell_{\nu-1}+1}^{\ell_0+\ell_1+\ldots+\ell_{\nu-1}+\ell_{\nu}}  \{ \In_j, \Mid_j , \Out_j\}$$ in stage $\nu$, for
$\nu=1,2,\ldots, s$, and $\mathsf{DC}$. We say that the cut thus defined is of type $(\ell_1,\ell_2,\ldots, \ell_s)$.
An example of cut with type  $(2,1,2)$ is shown in Fig.~\ref{fig:cut}.

\begin{figure}
\centering
\includegraphics[width=2.9in]{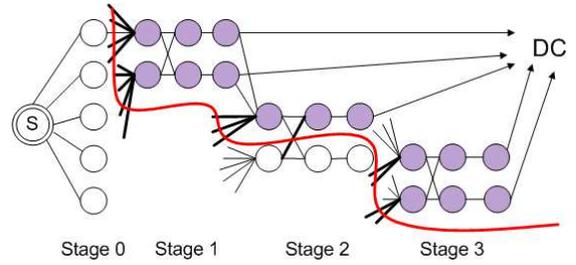}
\caption{A sample cut through the information flow graph.} \label{fig:cut}
\end{figure}

We claim that the capacity of an $(\mathsf{S},\DC)$-cut of type $(\ell_1,\ell_2,\ldots, \ell_s)$ can be as small as
\begin{equation}
 \sum_{\nu=1}^s \Big[ \ell_\nu(d-\sum_{j=1}^{\nu-1} \ell_j) \beta_1 + \ell_\nu(r-\ell_\nu)\beta_2 \Big]. \label{eq:cut}
\end{equation}
In stage 1, there are $d\ell_1$, each with capacity $\beta_1$, terminating at $\In_1$, $\In_2,\ldots, \In_{\ell_1}$.

There are $\ell_1 d$ edges, each  with capacity $\beta_1$, terminating at $\In_1$, $\In_2,\ldots, \In_{\ell_1}$ in stage 1.  Also, there are $\ell_1 (r-\ell_1)$ edges, each with capacity $\beta_2$, terminating at $\Mid_1$, $\Mid_2,\ldots, \Mid_{\ell_1}$.  Hence, a total of $\ell_1d\beta_1 + \ell_1(r-\ell_1) \beta_2$ are contributed to the summation in~\eqref{eq:cut}. This is the summand corresponding to $\nu=1$ in~\eqref{eq:cut}.

For the second group of $\ell_2$ storage nodes in stage $2$, there may be $\ell_1$ links from the first group of storage nodes, which are not counted in the capacity of the cut. The sum of capacities of edges terminating at the ``in'' vertices in $\bar{\mathcal{U}}$ in stage~2 could be as small as $\ell_2(d-\ell_1)\beta_1$. Together with the sum of capacities of the edges terminating at the ``mid'' vertices, a total of
$\ell_2(d-\ell_1)\beta_1 + \ell_2(r-\ell_2) \beta_2$ are contributed to~\eqref{eq:cut}. This is the second summand in~\eqref{eq:cut}.
The rest of the summands can be derived similarly. This finishes the proof of the claim.

For a data file of size $B$, we should be able to construct a flow of value at least $B$. Hence $B$ is less than or equal to~\eqref{eq:cut} for all type $(\ell_1,\ell_2,\ldots, \ell_s)$ with $0< \ell_\nu \leq r$ for all $\nu=1,2,\ldots, s$, and $\ell_1+\ldots +\ell_s=k$. After some algebraic manipulations, we have the following upper bound on $B$,
\begin{equation}
B \leq dk \beta_1 + rk\beta_2 - \beta_1 \sum_{i=1}^s \sum_{j=1}^{i-1} \ell_i \ell_j - \beta_2 \sum_{i=1}^s \ell_i^2.
\label{eq:cut2}
\end{equation}
We note that if we substitute $\beta_1$ and $\beta_2$ by  $2B/(k(2d+r-k))$ and $\beta_2 = B/(k(2d+r-k))$ respectively, then we have equality  in~\eqref{eq:cut2}.

We finish the proof by considering the two cases.

{\em Case 1: $k \leq r$.}  Consider the cut of type $(\underbrace{1,1,\ldots, 1}_k)$. From~\eqref{eq:cut2}, we obtain
\begin{equation}
B \leq  (dk - k(k-1)/2) \beta_1 + (rk - k)\beta_2
\label{eq:cutA}
\end{equation}
From the cut  of type $(k,0,0,\ldots, 0)$, we have the following constraint,
\begin{equation}
B \leq  dk\beta_1 + (r-k)k \beta_2 \label{eq:cutB}
\end{equation}
If we multiply~\eqref{eq:cutA} by $2d$, multiply~\eqref{eq:cutB} by $r-1$, and add the two resulting inequalities, we get
\[
(2d+r-1) B \leq k(2d+r-k)(d\beta_1 + (r-1)\beta_2).
\]
This proves the lower bound in~\eqref{eq:maxflow} in Case~1.

To see that the lower bound can be met only when $\beta_1$ and $\beta_2$ are specified as in the theorem, we notice that \eqref{eq:cutA} and \eqref{eq:cutB} define an unbounded polyhedral region in the $\beta_1$-$\beta_2$ plane, with~\eqref{eq:vertex} as a vertex point.
If we want to minimize the objective function $d\beta_1 + (r-1)\beta_2$ over all point $(\beta_1,\beta_2)$ in this region, the optimal point is precisely the point given in~\eqref{eq:vertex}.

{\em Case 2:  $ k>r$.}
Consider a cut of type $
 (\underbrace{r,r,\ldots, r}_{a},b)$,
where $a=\lfloor k/r \rfloor$ and $b=k-ra$. The upper bound of $B$ in~\eqref{eq:cut2} becomes
\begin{equation}
B \leq \Big[dk - \frac{r^2a(a-1)}{2}-abr \Big] \beta_1 + (rk-ar^2-b^2) \beta_2
\label{eq:cutC}
\end{equation}
Together with the constraint obtained from a cut of type $(1,1,\ldots, 1)$, we set up a linear program and minimize $d \beta_1 + (r-1) \beta_2$
over all $(\beta_1,\beta_2)$ satisfying the inequalities in~\eqref{eq:cutA} and~\eqref{eq:cutC}.

Let $L_1$ be the straight line in the $\beta_1$-$\beta_2$ plane by setting the inequality in~\eqref{eq:cutA} to equality. Let $L_2$ be the straight line consisting of point $(\beta_1,\beta_2)$ satisfying~\eqref{eq:cutC} with equality. We can verify that the intersection of $L_1$ and $L_2$ is the point in~\eqref{eq:vertex}.

We investigate the slopes of $L_1$ and $L_2$.
The slope of $L_1$ is equal to $-(d-(k-1)/2)/(r-1)$, which is larger than the slope of the objective function, $-d/(r-1)$. For the slope of $L_2$, we first check that
\begin{align*}
& dk - \frac{r^2a(a-1)}{2} -abr > db\\
\Longleftarrow\ & dk - {r^2a(a-1)} -abr > db\\
\Longleftrightarrow\ & d > r(a-1) + b
\ \Longleftrightarrow\  d > k-r.
\end{align*}
We have used several times that $k=ar+b$ in the above derivation.
The last line holds by the assumptions $d \geq k$ and $r\geq 2$. Therefore,
\[
-\frac{dk - \frac{r^2a(a-1)}{2} -abr}{rk-ar^2-b^2} < \frac{db}{rk-ar^2-b^2} = - \frac{d}{r-1}.
\]
The slope of $L_2$ is strictly less than the slope of the objective function. Thus, the optimal point of the linear programming problem is the vertex in~\eqref{eq:vertex}. This completes the proof of Case~2.
\end{proof}

Note: Theorem~\ref{thm:maxflow} is obtained independently in~\cite{KSS}.

We can now show that the regenerating code discussed in Section~\ref{sec:example} is optimal, in the sense that given the parameters $B$, $d$, $k$ and $r$, the repair-bandwidth matches the lower bound in Theorem~\ref{thm:maxflow}. We have $B=8$ and $d=k=r=2$ in the example. From Theorem~\ref{thm:maxflow}, the repair-bandwidth cannot be less than $8\frac{2\cdot 2+2-1}{2(2\cdot 2+2-2)} = 5$.  We have shown in Section~\ref{sec:example} that the repair process requires exactly 5 packet transmissions per failed node, and therefore matches the optimal value.

For non-cooperative or one-by-one repair, it is proved in~\cite{DGWR2010} that the minimum repair-bandwidth per failed node is $2dB/(k(2d+1-k))$,
which turns out to be the same as the left hand side of~\eqref{eq:maxflow} with $r$ set to 1. If we apply a non-cooperative regenerating code to a distributed storage system with parameters as in Section~\ref{sec:example}, the minimum repair-bandwidth is
$\frac{2\cdot 2 \cdot 8}{2(2\cdot 2+1-2)} = 5.333$. From this simple example, we can see that repair-bandwidth can be further reduced if some data exchange of data among the newcomers is allowed.

The lower bound of repair-bandwidth in Theorem~\ref{thm:maxflow} in fact holds via random linear coding with field size large enough. The tightness of the lower bound is established in~\cite{ShumHu11a},
by showing the existence of MBCR codes which match the lower bound. Thus, the minimum repair-bandwidth for MBCR is indeed equal to $B(2d+r-1)/(k(2d+r-k))$.

\section{An Explicit Construction of a Family of Optimal MBCR Codes}
\label{sec:explicit}
We construct in this section a family of exact MBCR code with parameters $d=k$ and $n=d+r$. In fact, the illustrated code in Section~\ref{sec:example} is a special case in this family.

The whole file is first divided into stripes. Each stripe consists of $B=k(2d+r-k)=kn$ data packets, considered as elements in $GF(q)$, where $q$ is a prime power.  In each stripe let the $kn$ data packets be $x_0$, $x_1, \ldots, x_{kn-1}$. We divide them into $n$ groups. The first group consists of $x_0, x_1,\ldots, x_{k-1}$, the second group consists of $x_{k}, x_{k+1}, \ldots, x_{2k-1}$, and so on. For notational convenience, we let $\mathbf{x}_j = [x_{(j-1)k} \ x_{(j-1)k+1}\ \ldots \ x_{(j-1)k+k-1}]$ be the vector of the  data packets in the $j$th group ($1\leq j \leq n$).

For $i=1,2,\ldots, n$, we construct the content of node $i$ as follows. We first put the $k$ data packets in the $i$-th group $\mathbf{x}_i$ into node $i$ and then $n-1$ parity-check packets
\[
\mathbf{x}_{i\oplus 1}\cdot \mathbf{v}_1,\
\mathbf{x}_{i\oplus 2}\cdot \mathbf{v}_2, \ \ldots,\
\mathbf{x}_{i\oplus (n-1)}\cdot \mathbf{v}_{n-1}
\]
into node $i$, where ``$\cdot$'' is the dot product of vectors and $\oplus$ is modulo-$n$ addition defined by
\[
 x \oplus y := \begin{cases}
 x+y & \text{ if } x+y \leq n, \\
 x+y-n & \text{ if } x+y > n.
 \end{cases}
\]
Here $\mathbf{v}_j$ ($j=1,2,\ldots, n-1$) are column vectors in a $k\times (n-1)$   generating matrix,  $\mathbf{G} = [\mathbf{v}_1\ \mathbf{v}_2\ \ldots\ \mathbf{v}_{n-1}]$, of a maximal-distance separable (MDS) code over $GF(q)$ of length $n-1$ and dimension $k$. By the defining property of MDS code, any $k$ columns of $\mathbf{G}$ are linearly independent of $GF(q)$.

As for the file reconstruction processing, suppose without loss of generality that a data collector  connects to nodes 1, $2, \ldots, k$.  The systematic packets $x_0, x_1, \ldots, x_{k^2-1}$ in the first $k$ groups can be downloaded directly, because they are stored in node 1 to node $k$ uncoded. The $j$th group of data packets ($j > k$) (the components in vector $\mathbf{x}_j$) can be reconstructed from
$\mathbf{x}_j \cdot \mathbf{v}_{j-1}$, $\mathbf{x}_j \cdot \mathbf{v}_{j-2}, \ldots,  \mathbf{x}_j \cdot \mathbf{v}_{j-k}$,
by the MDS property.
A data collector connecting to any other $k$ storage nodes can decode similarly.

As for the cooperative repair processing, suppose without loss of generality that nodes $k+1$ to $n$ fail at the same time. The repair process proceeds as follows.

\begin{description}
\item[\hspace{-10pt} Step 1:] For $i=1,2,\ldots, k$, node $i$ computes $\mathbf{x}_i \cdot \mathbf{v}_{n+i-j}$ and sends it
to newcomer $j$, for $j=k+1, k+2, \ldots, n$. 

\item[\hspace{-10pt} Step 2:] For $j=k+1, k+2, \ldots,n $, newcomer $j$ downloads $k$ packets $\mathbf{x}_{j}\cdot \mathbf{v}_{j-1}$, $\mathbf{x}_{j}\cdot \mathbf{v}_{j-2}, \ldots, \mathbf{x}_{j}\cdot \mathbf{v}_{j-k}$
    from nodes 1 to $k$.

\item[\hspace{-10pt} Step 3:] For $j=k+1, k+2, \ldots,n $, newcomer $j$ can solve for the systematic packets in $\mathbf{x}_j$. Then node $j$ sends $\mathbf{x}_j \cdot \mathbf{v}_{j+i}$ to node $n-i+1$, for $i=1,2,\ldots, n-j$, and sends  $\mathbf{x}_j \cdot \mathbf{v}_{i}$ to node $j-i$, for $i=1,2,\ldots, j-k-1$.
 \label{step3}
\end{description}
In steps 1 and 2,  a total of $2k(n-k)=2kr$ packets are transmitted. In step~3, each newcomer transmits $r-1$ packets. The total number of packets required in the whole repair process is
$2kr + r(r-1) = r(2d+r-1)$.
The number of packets per failed node is therefore $2d+r-1$. According to Theorem~\ref{thm:maxflow}, the repair-bandwidth is no less than
$$ B\frac{2d+r-1}{k(2d+r-k)} =  2d+r-1.$$
Thus, this regenerating code is optimal.

{\em Remark}: If $n = q+2$ for some prime power $q$, we can use an extended Reed-Solomon (RS) code of length $q+1$ in the construction. The alphabet size could be as small as $n-2$. We refer the reader to~\cite{Roth} for the construction of extended RS code.

{\em Example:}
An example for $n=5$, $d=k=3$ and $r=2$ is shown in Fig.~\ref{fig:construction}. A stripe of file data is divided into 15 packets $x_0$, $x_1, \ldots, x_{14}$. Let $q=2$ and $\mathbf{G}$ be the generating matrix
\[
 \mathbf{G}  = [\mathbf{v}_1\ \mathbf{v}_2\ \mathbf{v}_3\ \mathbf{v}_4] = \begin{bmatrix}
 1&1&0&0 \\
 1&0&1&0 \\
 1&0&0&1 \\
 \end{bmatrix}
\]
of a triply-extended Reed-Solomon code over $GF(2)$~\cite{Roth}. The $i$th row of the array in  Fig.~\ref{fig:construction} indicates the content of node~$i$. For example,
node 4 stores six systematic packets, $x_9$, $x_{10}$, $x_{11}$ in $\mathbf{x}_4$, $\mathbf{x}_1 \cdot \mathbf{v}_2 = x_0$, $\mathbf{x}_2\cdot \mathbf{v}_3 = x_4$, $\mathbf{x}_3 \cdot \mathbf{v}_4 = x_{8}$, and one parity-check packet $\mathbf{x}_5 \cdot \mathbf{v}_1 = x_{12}+x_{13}+x_{14}$.

\begin{figure}
\[ {\small
\begin{array}{|c||c|c|c|c|c|} \hline
 \text{Node 1}  & \mathbf{x}_1 & \mathbf{x}_2 \cdot \mathbf{v}_1 & \mathbf{x}_3 \cdot \mathbf{v}_2 &\mathbf{x}_4 \cdot \mathbf{v}_3 & \mathbf{x}_5 \cdot \mathbf{v}_4\\ \hline
\text{Node 2} &  \mathbf{x}_1 \cdot \mathbf{v}_4 &\mathbf{x}_2 & \mathbf{x}_3 \cdot \mathbf{v}_1 & \mathbf{x}_4 \cdot \mathbf{v}_2 &\mathbf{x}_5 \cdot \mathbf{v}_3 \\ \hline
\text{Node 3} &  \mathbf{x}_1 \cdot \mathbf{v}_3 & \mathbf{x}_2\cdot \mathbf{v}_4 & \mathbf{x}_3  & \mathbf{x}_4 \cdot \mathbf{v}_1 &\mathbf{x}_5 \cdot \mathbf{v}_2 \\ \hline
\text{Node 4} &  \mathbf{x}_1 \cdot \mathbf{v}_2 & \mathbf{x}_2\cdot \mathbf{v}_3 & \mathbf{x}_3 \cdot \mathbf{v}_4 & \mathbf{x}_4  &\mathbf{x}_5 \cdot \mathbf{v}_1 \\ \hline
\text{Node 5} &  \mathbf{x}_1 \cdot \mathbf{v}_1 & \mathbf{x}_2\cdot \mathbf{v}_2 & \mathbf{x}_3 \cdot \mathbf{v}_3 & \mathbf{x}_4 \cdot \mathbf{v}_4  &\mathbf{x}_5  \\ \hline
\end{array} }
\]
\caption{An example for exact MBCR code for $n=5, d=k=3$ and $r=2$.}
\label{fig:construction}
\end{figure}

Suppose that nodes 4 and 5 fail. In the first step, node 1 sends $\mathbf{x}_1\cdot \mathbf{v}_2$ and $\mathbf{x}_1 \cdot \mathbf{v}_1$ to newcomers 4 and 5 respectively.
Similarly, node 2 sends $\mathbf{x}_2 \cdot \mathbf{v}_3$ and $\mathbf{x}_2 \cdot \mathbf{v}_2$, and
node 3 sends $\mathbf{x}_3 \cdot \mathbf{v}_4$ and $\mathbf{x}_3 \cdot \mathbf{v}_3$.
In the second step, node 1 transmits $\mathbf{x}_4 \cdot \mathbf{v}_3$ and $\mathbf{x}_5\cdot \mathbf{v}_4$ to newcomers 4 and 5 respectively. Likewise, node 2 transmits $\mathbf{x}_4 \cdot \mathbf{v}_2$ and $\mathbf{x}_5  \cdot \mathbf{v}_3$, and node 3 transmits $\mathbf{x}_4  \cdot \mathbf{v}_1$ and $\mathbf{x}_5  \cdot \mathbf{v}_2$. In the third step, newcomer 4 reconstructs $\mathbf{x}_4$, and sends $\mathbf{x}_4 \cdot \mathbf{v}_4$ to newcomer~5.
Also, newcomer 5 reconstructs  $\mathbf{x}_5$, and sends $\mathbf{x}_5 \cdot \mathbf{v}_1$ to newcomer~4. Lastly, the lost packets in nodes 4 and 5 are regenerated in newcomer 4 and 5. The total number of packet transmissions in the whole repair process is equal to 14. The repair-bandwidth per failed node is~7. It matches the theoretic lower bound $15(2\cdot 3+2-1)/(3(2\cdot 3+2-3)) = 7$.

\section{Conclusion}

We give a construction of a family of exact and optimal MBCR codes for $d=k$ and $n=d+r$. The constructed regenerating code has the advantage of being a systematic code.  For example, if we want to look at the content of one particular packet, we only need to contact the node which has a copy of this packet and download the packet directly. Another advantage of this construction is that the requirement of finite field size grows linearly as a function of the number of storage nodes.




\end{document}